\documentclass[10pt,dvipdfmx]{article}

\usepackage{latexsym,amsfonts,amsmath,amssymb,mathrsfs,url,amsthm}
\usepackage[dvipdfmx]{color,graphicx}
\newtheorem{theorem}{Theorem}
\newtheorem{lemma}{Lemma}

\newtheorem{algorithm}{Algorithm}
\newtheorem{remark}{Remark}

\usepackage[dvipdfm,
linkbordercolor={1 0 0},
citebordercolor={0 1 0},
urlbordercolor={0 1 1}]{hyperref}

\newcommand{\EE}{\mathbb{E}}

\newcommand{\VV}{\mathbb{V}}
\newcommand{\RR}{\mathbb{R}}
\newcommand{\ZZ}{\mathbb{Z}}
\newcommand{\bszero}{\boldsymbol{0}}
\newcommand{\tg}[1]{{\color{black}{#1}}}
\allowdisplaybreaks

\title{Multilevel Monte Carlo estimation of expected information gains}
\author{Takashi Goda\thanks{School of Engineering, University of Tokyo, 7-3-1 Hongo, Bunkyo-ku, Tokyo 113-8656, Japan ({\tt goda@frcer.t.u-tokyo.ac.jp}, {\tt hironaka-tomohiko@g.ecc.u-tokyo.ac.jp}, {\tt takeru-iwamoto735@g.ecc.u-tokyo.ac.jp})}, Tomohiko Hironaka\footnotemark[1], Takeru Iwamoto\footnotemark[1]}
\date{\today}

\begin{document}
\maketitle
\begin{abstract}
The expected information gain is an important quality criterion of Bayesian experimental designs, which measures how much the information entropy about uncertain quantity of interest $\theta$ is reduced on average by collecting relevant data $Y$. However, estimating the expected information gain has been considered computationally challenging since it is defined as a nested expectation with an outer expectation with respect to $Y$ and an inner expectation with respect to $\theta$. In fact, the standard, nested Monte Carlo method requires a total computational cost of $O(\varepsilon^{-3})$ to achieve a root-mean-square accuracy of $\varepsilon$. In this paper we develop an efficient algorithm to estimate the expected information gain by applying a multilevel Monte Carlo (MLMC) method. To be precise, we introduce an antithetic MLMC estimator for the expected information gain and provide a sufficient condition on the data model under which the antithetic property of the MLMC estimator is well exploited such that optimal complexity of $O(\varepsilon^{-2})$ is achieved. Furthermore, we discuss how to incorporate importance sampling techniques within the MLMC estimator to avoid arithmetic underflow. Numerical experiments show the considerable computational cost savings compared to the nested Monte Carlo method for a simple test case and a more realistic pharmacokinetic model.
\\
\textbf{Keywords:} expected information gain, Bayesian experimental design, multilevel Monte Carlo
\end{abstract}

\section{Introduction}
The motivation for this research comes from construction of optimal Bayesian experimental designs, where the so-called \emph{expected information gain} has been often employed as a quality criterion of experimental designs, see for instance \cite{L56,CV95,Ryan03,HM13,LSTW13,BDELT18}. Let $\theta$ be a (possibly multi-dimensional) random variable which represents the uncertain quantity of interest. By collecting relevant data $Y_\xi$ (which is again possibly multi-dimensional) through carrying out some experiments under an experimental setup $\xi$, we expect that the uncertainty of $\theta$ can be reduced. As originally advocated in \cite{L56}, here we measure the uncertainty of $\theta$ by its information entropy. The aim of Bayesian experimental designs is to find an optimal experimental setup $\xi^*$ which maximizes the expected information gain, that is, the expected amount of the information entropy reduction about $\theta$. \tg{If $\xi$ lives in a continuous space, one may want to evaluate the derivative of the expected information gain with respect to $\xi$, instead of the expected information gain itself, in order to search for a maximizer $\xi^*$. If not, however, accurate evaluation of the expected information gain for given $\xi$ plays an essential role in constructing optimal Bayesian experimental designs.}

In what follows, we give a formal definition of the expected information gain for a particular experimental setup $\xi$. The information entropy of $\theta$ before collecting data $Y_\xi$ is given by
\[ -\EE_\theta[\log p(\theta)], \]
where $p(\theta)$ denotes the prior probability density function of $\theta$. On the other hand, after collecting data $Y_\xi$, the conditional information entropy of $\theta$ is
\[ -\EE_{\theta|Y_\xi}[\log p(\theta\,|\,Y_\xi)], \]
where $p(\theta\,|\, Y_\xi)$ denotes the posterior probability density function of $\theta$ given $Y_\xi$. Note that the expectation is now taken with respect to $p(\theta\,|\, Y_\xi)$ instead of $p(\theta)$. Thus the expected conditional information entropy of $\theta$ by collecting data $Y_\xi$ is
\[ \EE_{Y_\xi}\left[-\EE_{\theta|Y_\xi}[\log p(\theta\,|\,Y_\xi)]\right]. \]
The expected information gain is defined by the difference
\begin{align} 
U_\xi & := -\EE_\theta[\log p(\theta)] - \EE_{Y_\xi}\left[-\EE_{\theta|Y_\xi}[\log p(\theta\,|\,Y_\xi)]\right] \notag \\
& = \EE_{Y_\xi}\left[  -\EE_{\theta\,|\, Y_\xi}[\log p(\theta)] + \EE_{\theta|Y_\xi}[\log p(\theta\,|\,Y_\xi)]\right] \notag \\
& = \EE_{Y_\xi}\EE_{\theta\,|\, Y_\xi}\left[ \log \frac{p(\theta\,|\, Y_\xi)}{p(\theta)}\right].\label{eq:eig}
\end{align}
This means that the expected information gain $U_\xi$ measures the average amount of the reduction of the information entropy about $\theta$ by collecting data $Y_\xi$. In \eqref{eq:eig}, the inner expectation appearing in the right-most side is nothing but the Kullback-Leibler divergence between $p(\theta)$ and $p(\theta\,|\, Y_\xi)$. In the context of Bayesian experimental designs, we claim that the data $Y_\xi$ with larger value of $U_\xi$ is more informative about $\theta$ and thus the corresponding experimental design $\xi$ is better. This is why the expected information gain is used as a quality criterion of experimental designs \cite{L56}.

Let us consider the following data model:
\begin{align}\label{eq:data_model}
 Y_\xi = g_\xi(\theta)+\epsilon,
\end{align}
where the function $g_\xi$ represents the deterministic part of the model response which depends on $\theta$ and $\xi$, and $\epsilon$ denotes the stochastic part of the model response, i.e, the measurement error. Typically $\epsilon$ is assumed to be zero-mean Gaussian with covariance matrix $\Sigma_\epsilon$. As considered in \cite{HM13,LSTW13,BDELT18}, this data model can be extended to allow the repetition of experiments as
\[ Y_\xi^{(i)} = g_\xi(\theta)+\epsilon^{(i)}\quad \text{for $i=1,\ldots,N_e$,} \]
where $N_e$ is the number of repetitive experiments and $\epsilon^{(i)}$ are independent and identically distributed (i.i.d.) measurement errors. However, this extended model can be easily rewritten into the form of \eqref{eq:data_model} by concatenating $Y_\xi=(Y_\xi^{(1)},\ldots,Y_\xi^{(N_e)})$, so that we stick to the original model \eqref{eq:data_model} in this paper. 

As an initial but crucial step toward an efficient construction of optimal Bayesian experimental designs, we develop an efficient Monte Carlo algorithm for estimating the expected information gain $U_\xi$ for a given experimental setup $\xi$ in this paper. \tg{Since we shall fix $\xi$ in the rest of this paper, we omit the subscript $\xi$ and simply write $g,Y,U$ instead of $g_\xi,Y_\xi,U_\xi$ when distinguishing different $\xi$'s is not important.} In the next section, we introduce the standard, nested Monte Carlo method as a classical algorithm to estimate $U$, and give a brief review of the relevant literature. Then in Section~\ref{sec:mlmc}, after introducing the concept of a multilevel Monte Carlo (MLMC) method, we construct an MLMC estimator for $U$ as an alternative, more efficient algorithm. We prove under a sufficient condition on the data model that the MLMC estimator can estimate $U$ with a root-mean-square accuracy $\varepsilon$ by the computational cost of optimal $O(\varepsilon^{-2})$. (Here and in what follows, the difference between the noise $\epsilon$ and the accuracy $\varepsilon$ should not be confused.) Recently in \cite{TGH17}, Tsilifis et al.\ considered a lower bound on the expected information gain as a criterion of experimental designs and showed that the same order of computational cost can be achieved by the standard Monte Carlo method to estimate it. Our proposal enables to estimate the expected information gain itself efficiently, which is the main contribution of this paper. Moreover we discuss how to incorporate importance sampling techniques within the MLMC estimator, \tg{which proves to be useful in some applications.} Numerical experiments in Section~\ref{sec:numer} confirm the considerable computational savings compared to the nested Monte Carlo method not only for a simple test case but also for a more realistic pharmacokinetic model adapted from \cite{RDTP14}. \tg{Section~\ref{sec:conclusion} concludes this paper with some remarks on future research directions.}

\section{Nested Monte Carlo}\label{sec:nmc}
The nested Monte Carlo (NMC) method is the most standard approach to estimate the expected information gain \cite{Ryan03,HM13,BDELT18,RCYWW18}. Given the data model \eqref{eq:data_model}, it is straightforward to generate i.i.d.\ random samples of $Y$ given a particular value of $\theta$ and also those of $Y$ itself. Besides, since $Y-g(\theta)$ follows the probability distribution of $\epsilon$, it is easy to compute $p(Y\,|\, \theta)$ for given $\theta$ and $Y$. On the other hand, it is usually hard to generate i.i.d.\ random samples of $\theta$ given a particular value of $Y$ and to compute $p(\theta\,|\, Y)$ and $p(Y)$ for given $\theta$ and $Y$. 

Based on this fact, we use Bayes' theorem
\[ p(\theta\,|\,Y) = \frac{p(\theta)p(Y\,|\,\theta)}{p(Y)} = \frac{p(\theta)p(Y\,|\,\theta)}{\EE_\theta[p(Y\,|\, \theta)]}, \]
to rewrite the expected information gain $U$, defined in \eqref{eq:eig}, into
\begin{align}
U & =  \EE_{Y}\EE_{\theta\,|\, Y}\left[ \log \frac{p(Y\,|\,\theta)}{\EE_\theta[p(Y\,|\, \theta)]}\right] \notag \\
& =  \EE_Y \EE_{\theta|Y}\left[\log p(Y\,|\,\theta) \right] -\EE_Y\left[\log \EE_\theta[p(Y\,|\, \theta)]\right] \notag \\ 
& =  \EE_{\theta}\EE_{Y|\theta}\left[\log p(Y\,|\,\theta) \right] -\EE_Y\left[\log \EE_\theta[p(Y\,|\, \theta)]\right] .\label{eq:eig2}
\end{align}
With this form of $U$, the NMC estimator for the expected information gain is given by
\begin{align}\label{eq:NMC_estimator}
 \frac{1}{N}\sum_{n=1}^{N}\left[ \log p(Y^{(n)}\,|\, \theta^{(n,0)})-\log\left( \frac{1}{M}\sum_{m=1}^{M}p(Y^{(n)}\,|\, \theta^{(n,m)})\right) \right],
\end{align}
for some $M,N>0$, where $\theta^{(n,0)},\theta^{(n,1)},\ldots,\theta^{(n,m)}$ denote i.i.d.\ random samples of $\theta$, and $Y^{(n)}$ denotes a random sample of $Y$ generated conditionally on $\theta^{(n,0)}$.

In \cite{Ryan03}, Ryan showed under some approximations that the bias and the variance of the NMC estimator are of $O(M^{-1})$ and of $O(N^{-1})$, respectively. Since the mean square error of the NMC estimator is given by the sum of the variance and the squared bias, $U$ can be estimated with a root-mean-square accuracy $\varepsilon$ by using $N=O(\varepsilon^{-2})$ and $M=O(\varepsilon^{-1})$ samples. Assuming that each computation of $g$, which is necessary for calculating $p(Y\,|\, \theta)$, can be performed with unit cost, the total computational cost is $N(M+1)=O(\varepsilon^{-3})$.

Much more recently, in \cite{BDELT18}, Beck et al.\ provided a thorough error analysis of the NMC estimator and derived the optimal allocation of $N$ and $M$ for a given $\varepsilon$. In fact, they considered the situation where $g$ cannot be computed exactly and only its discretized approximation $g_h$ with a mesh discretization parameter $h>0$ is available. Here $g_h$ approaches to $g$ as $h$ gets smaller, but at the same time, the computational cost of $g_h$ increases. Therefore, their optimization deals with not only the number of samples $N$ and $M$ but also the parameter $h$. In this paper, we assume that $g$ can be computed exactly, so that dealing with such situations is left open for future works, see Section~\ref{sec:conclusion}.

More importantly, Beck et al.\ incorporated importance sampling based on the Laplace approximation from \cite{LSTW13} within the NMC estimator. This approach is quite useful in reducing the number of inner samples $M$ substantially and also in mitigating the risk of arithmetic underflow. When $p(Y\,|\, \theta)$ (as a function of $\theta$ for a fixed $Y$) is highly concentrated around a certain value of $\theta$, the Monte Carlo estimate of the inner expectation
\[ \frac{1}{M}\sum_{m=1}^{M}p(Y^{(n)}\,|\, \theta^{(n,m)}), \]
appearing in \eqref{eq:NMC_estimator} can be numerically zero. Taking the logarithm of 0 of course returns error. This can happen in practice especially for small $M$. Therefore, applying a change of measure such that most of the samples of $\theta$ are concentrated properly depending on $Y^{(n)}$ is desirable, which is exactly what the Laplace-based importance sampling aims to do. We note, however, that using importance sampling does not improve the order of computational complexity, so that the necessary cost of $O(\varepsilon^{-3})$ remains unchanged.

\section{Multilevel Monte Carlo}\label{sec:mlmc}
\subsection{Basic theory of MLMC}
In order to \tg{reduce} the necessary computational cost to estimate $U$ from $O(\varepsilon^{-3})$ to $O(\varepsilon^{-2})$, we consider applying a multilevel Monte Carlo (MLMC) method \cite{G08,G15}. The MLMC method has already been applied to estimate nested expectations of the form
\[ \EE\left[f \left( \EE[g(X,Y)\,|\, Y]\right) \right], \]
for independent random variables $X$ and $Y$, where an outer expectation is taken with respect to $Y$ and an inner one is taken with respect to $X$, see \cite{BHR15,G15,GG19,GH19}. In particular, the case where $f$ is \tg{twice differentiable} has been briefly discussed in \cite[Section~9]{G15} \tg{based on a Taylor series expansion of $f$}. In this paper we make a rigorous argument when $f$ is a logarithmic function, \tg{for which the remainder term of the Taylor expansion has to be carefully dealt with}.

Before introducing an MLMC estimator for the expected information gain, we give an overview of the MLMC method. Let $P$ be a random output variable which cannot be sampled exactly, and let $P_0,P_1,\ldots $ be a sequence of random variables which approximate $P$ with increasing accuracy but also with increasing cost. The problem here is to estimate $\EE[P]$ efficiently. 

For $L\in \ZZ_{>0}$ we have the following telescoping sum
\begin{align}\label{eq:tele_sum}
 \EE[P_L] = \EE[P_0]+\sum_{\ell=1}^{L}\EE[P_\ell-P_{\ell-1}].
\end{align}
The standard Monte Carlo method estimates the left-hand side directly by
\begin{align}\label{eq:SMC}
 Z_{\text{MC}}=\frac{1}{N}\sum_{i=1}^{N}P_L^{(i)}. 
\end{align}
The mean square error of $Z_{\text{MC}}$ is given by the sum of variance \tg{$\VV$} and squared bias:
\begin{align}\label{eq:SMC_mse}
\EE[( Z_{\text{MC}}-\EE[P])^2] = \frac{\VV[P_L]}{N}+\left( \EE[P_L-P]\right)^2.
\end{align}
The MLMC method, on the other hand, independently estimates each term on the right-hand side of \eqref{eq:tele_sum}. In general, if we have a sequence of random variables $Z_0,Z_1,\ldots$ which satisfy $\EE[Z_0]=\EE[P_0]$ and $\EE[Z_\ell]=\EE[P_\ell-P_{\ell-1}]$ for $\ell\in \ZZ_{>0}$, the MLMC estimator is given by
\begin{align}\label{eq:MLMC}
Z_{\text{MLMC}}=\sum_{\ell=0}^{L}\frac{1}{N_\ell}\sum_{i=1}^{N_\ell}Z_\ell^{(i)}.
\end{align}
The mean square error of $Z_{\text{MLMC}}$ is 
\begin{align}\label{eq:MLMC_mse}
 \EE[( Z_{\text{MLMC}}-\EE[P])^2] = \sum_{\ell=0}^{L}\frac{\VV[Z_\ell]}{N_\ell}+\left( \EE[P_L-P]\right)^2.
\end{align}
For the same underlying stochastic sample, $P_\ell$ and $P_{\ell-1}$ can be well correlated and thus $\VV[Z_\ell]$ is expected to get smaller as the level $\ell$ increases. This means that, in order to estimate $\EE[Z_\ell]$ efficiently, the necessary number of samples $N_\ell$ decreases as $\ell$ increases, and, as a consequence, most of the number of samples are allocated on smaller levels for estimating $\EE[P_L]$. Since the cost for each computation of $Z_\ell$ is assumed to be cheaper for smaller $\ell$, the overall computational cost can be significantly reduced compared to the standard Monte Carlo method.

In his seminal work \cite{G08}, Giles made this observation explicit as follows, see also a recent review  \cite{G15}:
\begin{theorem}\label{thm:basic_MLMC}
Let $P$ be a random variable and let $P_\ell$ denote the corresponding level $\ell$ approximation of $P$. If there exist independent random variables $Z_\ell$ with expected cost $C_\ell$ and variance $V_\ell$, and positive constants $\alpha,\beta,\gamma,c_1,c_2,c_3$ such that $\alpha\geq \min(\beta,\gamma)/2$ and
\begin{enumerate}
\item (decay of bias) $|\EE[P_\ell-P]| \leq c_12^{-\alpha \ell}$,
\item (proper coupling) $\displaystyle \EE[Z_\ell] =\begin{cases} \EE[P_0] & \ell=0, \\ \EE[P_\ell-P_{\ell-1}] & \ell>0, \end{cases}$
\item (decay of variance) $V_\ell \leq c_22^{-\beta \ell}$,
\item (growth of cost) $C_\ell \leq c_32^{\gamma \ell}$,
\end{enumerate}
then there exists a positive constant $c_4$ such that for any $\varepsilon<\tg{\exp(-1)}$ there are $L$ and $N_\ell$ for which the MLMC estimator \eqref{eq:MLMC} has a mean square error less than $\varepsilon^2$ with a computational complexity $C$ with bound
\[ \EE[C]\leq \begin{cases} c_4\varepsilon^{-2} & \beta>\gamma, \\ c_4\varepsilon^{-2}(\log \varepsilon)^2 & \beta=\gamma, \\ c_4\varepsilon^{-2-(\gamma-\beta)/\alpha} & \beta<\gamma. \end{cases} \]
\end{theorem}

\begin{remark}\label{rem:smc}
As discussed for instance in \cite[Section~2.1]{GG19}, a computational complexity for the standard Monte Carlo estimator to have a mean square error less than $\varepsilon^2$ is of $O(\varepsilon^{-2-\gamma/\alpha})$. Thus regardless of the values of $\beta$ and $\gamma$, the MLMC estimator has an asymptotically better complexity bound than the standard Monte Carlo estimator.
\end{remark}

\subsection{MLMC estimator for expected information gains}
Here we introduce an MLMC estimator for the expected information gain. First let us define a random output variable
\[ P := \log p(Y\,|\,\theta) - \log \EE_\theta[p(Y\,|\, \theta)], \]
where $Y$ is distributed conditionally on the random variable $\theta$ of the first term. It is obvious that $P$ cannot be computed exactly because of the expectation $\EE_\theta[p(Y\,|\, \theta)]$ appearing in the second term. However, we can introduce a sequence of approximations $P_0,P_1,\ldots$ of $P$ with increasing accuracy but also with increasing cost as follows:
\begin{align*}
 P_\ell & = \log p(Y \,|\, \theta)-\log\left( \frac{1}{M_\ell}\sum_{m=1}^{M_\ell}p(Y\,|\, \theta^{(m)})\right) \\
 & =: \log p(Y \,|\, \theta) - \log \overline{p(Y \,|\, \cdot )}^{M_\ell},
\end{align*}
for an increasing sequence $M_0<M_1<\ldots$ such that $M_\ell\to \infty$ as $\ell\to \infty$. That is, $P_\ell$ is the standard Monte Carlo estimator of $P$ using $M_\ell$ random samples of $\theta$. Thus we have $\lim_{\ell\to \infty}\EE[P_\ell]=\EE[P]$. Note that the standard, nested Monte Carlo estimator \eqref{eq:NMC_estimator} is essentially the same as \eqref{eq:SMC} with $P_L$ given as above for a fixed $L$. 

In what follows, let $M_\ell:=M_02^{\ell}$ for some $M_0\in \ZZ_{>0}$ for all $\ell\geq 0$, i.e., we consider a geometric progression for $M_\ell$. Then a sequence of corrections $Z_0,Z_1,\ldots$ is defined as follows: $Z_0$ is the same as $P_0$, given by
\[ Z_0 = \log p(Y \,|\, \theta) - \log \overline{p(Y \,|\, \cdot )}^{M_0}. \]
For $\ell>0$, the simplest form of $Z_\ell$ is 
\[ Z_\ell = P_\ell-P_{\ell-1}=\log \overline{p(Y \,|\, \cdot )}^{M_{\ell-1}}-\log \overline{p(Y \,|\, \cdot )}^{M_\ell}, \]
where the first $M_{\ell-1}$ random samples of $\theta$ used in the second term is also used in the first term.
However, according to \cite{GS14,BHR15,G15,GG19}, we can consider a better ``tight coupling'' of $P_\ell$ and $P_{\ell-1}$. Namely, the set of $M_02^{\ell}$ random samples of $\theta$ used to compute $P_{\ell}$ is divided into two disjoint sets of $M_02^{\ell-1}$ samples to compute two realizations of $P_{\ell-1}$, denoted by $P_{\ell-1}^{(a)}$ and $P_{\ell-1}^{(b)}$, respectively. This way we define $Z_\ell$ by
\begin{align}
 Z_\ell & = P_\ell-\frac{1}{2}\left[ P_{\ell-1}^{(a)}+P_{\ell-1}^{(b)}\right] \notag \\
 & = \log p(Y\,|\,\theta) - \log \overline{p(Y \,|\, \cdot )}^{M_02^{\ell}} \notag \\
 & \quad -\frac{1}{2}\left[ \log p(Y\,|\,\theta) - \log \overline{p(Y \,|\, \cdot )}^{(a)}+\log p(Y\,|\,\theta) - \log \overline{p(Y \,|\, \cdot )}^{(b)} \right] \notag \\
 & = \frac{1}{2}\left[ \log \overline{p(Y \,|\, \cdot )}^{(a)}+\log \overline{p(Y \,|\, \cdot )}^{(b)}\right] - \log \overline{p(Y \,|\, \cdot )} ,\label{eq:MLMC_correction}
\end{align}
where 
\begin{itemize}
\item $\overline{p(Y \,|\, \cdot )}$ denotes an average of $p(Y\,|\, \theta)$ over $M_02^{\ell}$ random samples of $\theta$ (note that we omit the superscript $M_02^{\ell}$ since it is clear from the level of $Z_{\ell}$);
\item $\overline{p(Y \,|\, \cdot )}^{(a)}$ denotes an average of $p(Y\,|\, \theta)$ over the first $M_02^{\ell-1}$ random samples of $\theta$ used in $\overline{p(Y \,|\, \cdot )}$;
\item $\overline{p(Y \,|\, \cdot )}^{(b)}$ denotes an average of $p(Y\,|\, \theta)$ over the second $M_02^{\ell-1}$ random samples of $\theta$ used in $\overline{p(Y \,|\, \cdot )}$,
\end{itemize}
for a randomly generated $Y$.
Because of the independence of $P_{\ell-1}^{(a)}$ and $P_{\ell-1}^{(b)}$, we see that $\EE[Z_\ell] = \EE[P_\ell-P_{\ell-1}]$. Moreover, it is important that the following ``antithetic'' property of $Z_{\ell}$ holds:
\begin{align}\label{eq:MLMC_antithetic}
 \frac{1}{2}\left[ \overline{p(Y \,|\, \cdot )}^{(a)}+\overline{p(Y \,|\, \cdot )}^{(b)}\right]=\overline{p(Y \,|\, \cdot )}.
\end{align}
Due to the concavity of $\log$, this $Z_\ell$ is always non-positive when $\ell\geq 1$. 

In this paper, we always consider the latter definition of $Z_\ell$ for $\ell>0$. Our MLMC estimator for the expected information gain is given by \eqref{eq:MLMC} for $L\in \ZZ_{>0}$ and $N_0,\ldots,N_L\in \ZZ_{>0}$ into which the above $Z_\ell$ is substituted. It is already clear from the construction of $Z_\ell$ that the parameter $\gamma$ in Theorem~\ref{thm:basic_MLMC} should be \tg{set to $\gamma=1$}.

\subsection{MLMC variance analysis}
In this subsection we prove $\beta>\gamma$ for $Z_\ell$ defined in \eqref{eq:MLMC_correction}, meaning that our MLMC estimator is in the first regime of Theorem~\ref{thm:basic_MLMC}, so that the total computational complexity is $O(\varepsilon^{-2})$.

In order to prove the main theorem below, we need the following result.

\begin{lemma}\label{lem:deviation}
Let $X$ be a random variable with zero mean, and let $\overline{X}_N$ be an average of $N$ i.i.d.\ samples of $X$. If $\EE[|X|^p]$ is finite for $p\geq 2$, there exists a constant $C_p$ depending only on $p$ such that
\[ \EE[|\overline{X}_N|^p] \leq C_p \frac{\EE[|X|^p]}{N^{p/2}} .\]
\end{lemma}

\begin{proof}
See \cite[Lemma~1]{GG19}.
\end{proof}

Now we prove:
\begin{theorem}\label{thm:main} If there exist $p, q> 2$ with $(p-2)(q-2)\geq 4$ such that 
\[ \EE_{\theta,Y}\left[\left|\frac{p(Y\,|\,\theta)}{p(Y)}\right|^p\right] <\infty\quad \text{and} \quad \EE_{\theta,Y}\left[\left|\log \frac{p(Y\,|\,\theta)}{p(Y)}\right|^q\right]<\infty, \]
respectively, we have
\[ \EE[|Z_\ell|] =O(2^{-\min(\frac{p(q-1)}{2q}, 1)\ell})\quad \text{and}\quad \VV[Z_{\ell}]=O(2^{- \min(\frac{p(q-2)}{2q}, 2) \ell}). \]
\end{theorem}

\begin{proof}
Using the antithetic property \eqref{eq:MLMC_antithetic} for a particular value of $Y$, we have
\begin{align*}
Z_\ell & = \frac{1}{2}\left[ \log \frac{\overline{p(Y \,|\, \cdot )}^{(a)}}{p(Y)}+\log \frac{\overline{p(Y \,|\, \cdot )}^{(b)}}{p(Y)}\right] - \log \frac{\overline{p(Y \,|\, \cdot )}}{p(Y)} \\
& \quad -\frac{1}{2}\left[ \frac{\overline{p(Y \,|\, \cdot )}^{(a)}}{p(Y)}+\frac{\overline{p(Y \,|\, \cdot )}^{(b)}}{p(Y)}\right]+\frac{\overline{p(Y \,|\, \cdot )}}{p(Y)}\\
& = \frac{1}{2}\left[ \log \frac{\overline{p(Y \,|\, \cdot )}^{(a)}}{p(Y)}-\frac{\overline{p(Y \,|\, \cdot )}^{(a)}}{p(Y)}+1\right] \\
& \quad +\frac{1}{2}\left[ \log \frac{\overline{p(Y \,|\, \cdot )}^{(b)}}{p(Y)}-\frac{\overline{p(Y \,|\, \cdot )}^{(b)}}{p(Y)}+1\right] - \left[ \log \frac{\overline{p(Y \,|\, \cdot )}}{p(Y)}-\frac{\overline{p(Y \,|\, \cdot )}}{p(Y)}+1\right] .
\end{align*}
Applying Jensen's inequality gives
\begin{align}
|Z_\ell|^2 & \leq \left| \log \frac{\overline{p(Y \,|\, \cdot )}^{(a)}}{p(Y)}-\frac{\overline{p(Y \,|\, \cdot )}^{(a)}}{p(Y)}+1\right|^2 \notag \\
& \quad +\left| \log \frac{\overline{p(Y \,|\, \cdot )}^{(b)}}{p(Y)}-\frac{\overline{p(Y \,|\, \cdot )}^{(b)}}{p(Y)}+1\right|^2 + 2\left| \log \frac{\overline{p(Y \,|\, \cdot )}}{p(Y)}-\frac{\overline{p(Y \,|\, \cdot )}}{p(Y)}+1\right|^2. \label{eq:bound_zl}
\end{align}
In what follows, we show a bound on the expectation of the last term of \eqref{eq:bound_zl}.

It is elementary to check that the following inequality holds
\[ |\log x-x+1|\leq  |x-1|^r \max\left(-\log x, 1\right), \]
for any $x>0$ and any $1\leq r\leq 2$. Thus it follows from H\"{o}lder's inequality that
\begin{align}
& \EE\left[ \left|\log \frac{\overline{p(Y \,|\, \cdot )}}{p(Y)}-\frac{\overline{p(Y \,|\, \cdot )}}{p(Y)}+1\right|^2 \right] \notag \\
& \leq \EE\left[ \left| \frac{\overline{p(Y \,|\, \cdot )}}{p(Y)}-1\right|^{2r} \left(\max\left( -\log \frac{\overline{p(Y \,|\, \cdot )}}{p(Y)},1\right) \right)^2 \right] \notag \\
& \leq \left( \EE\left[ \left| \frac{\overline{p(Y \,|\, \cdot )}}{p(Y)}-1\right|^{2sr}\right] \right)^{1/s}  \left( \EE\left[\left(\max\left( -\log \frac{\overline{p(Y \,|\, \cdot )}}{p(Y)},1\right) \right)^{2t} \right]\right)^{1/t},\label{eq:holder_bound}
\end{align}
for any H\"{o}lder conjugates $s,t\geq 1$ such that $1/s+1/t=1$.

For the first factor of \eqref{eq:holder_bound}, we recall that $\overline{p(Y \,|\, \cdot )}$ is an unbiased Monte Carlo estimate of $p(Y)$ using $M_02^{\ell}$ samples of $\theta$. Hence, as long as $2sr\leq p$, it follows from Lemma~\ref{lem:deviation} that 
\[ \EE\left[ \left| \frac{\overline{p(Y \,|\, \cdot )}}{p(Y)}-1\right|^{2sr}\right] \leq \frac{C_{2sr}}{(M_02^{\ell})^{sr}}\EE\left[ \left| \frac{p(Y\,|\, \theta)}{p(Y)}-1\right|^{2sr}\right] .\]
For the second factor of \eqref{eq:holder_bound}, we recall that the function $f(x)=\max\left(-\log x, 1\right)>0$ is convex. Thus, applying Jensen's inequality twice, we have
\begin{align*}
\left(\max\left( -\log \frac{\overline{p(Y \,|\, \cdot )}}{p(Y)},1\right) \right)^{2t} & \leq \left(\frac{1}{M_02^{\ell}}\sum_{m=1}^{M_02^{\ell}}\max\left( -\log \frac{p(Y \,|\, \theta^{(m)})}{p(Y)}, 1\right) \right)^{2t} \\
& \leq \frac{1}{M_02^{\ell}}\sum_{m=1}^{M_02^{\ell}}\left( \max\left( -\log \frac{p(Y \,|\, \theta^{(m)})}{p(Y)}, 1\right) \right)^{2t} \\
& \leq \frac{1}{M_02^{\ell}}\sum_{m=1}^{M_02^{\ell}}\left( \left|\log \frac{p(Y \,|\, \theta^{(m)})}{p(Y)}\right|^{2t}+1\right) .
\end{align*}
Thus we obtain
\[ \EE\left[\left(\max\left( -\log \frac{\overline{p(Y \,|\, \cdot )}}{p(Y)},1\right) \right)^{2t} \right] \leq \EE\left[ \left|\log \frac{p(Y \,|\, \theta)}{p(Y)}\right|^{2t}\right]  +1, \]
as long as $2t\leq q$. The H\"{o}lder conjugates $s$ and $t$ and the exponent $r$ can be chosen as
\[ s=\frac{q}{q-2},\quad t=\frac{q}{2}\quad \text{and}\quad r=\min\left( \frac{p(q-2)}{2q}, 2\right), \]
respectively. Here the assumption $(p-2)(q-2)\geq 4$ of the theorem is required to ensure $r\geq 1$. Altogether the expectation of the last term of \eqref{eq:bound_zl} is bounded above by
\begin{align*}
& \EE\left[ \left|\log \frac{\overline{p(Y \,|\, \cdot )}}{p(Y)}-\frac{\overline{p(Y \,|\, \cdot )}}{p(Y)}+1\right|^2 \right] \\
& \leq \frac{C_{2sr}^{1/s}}{(M_02^{\ell})^{r}}\left(\EE\left[ \left| \frac{p(Y\,|\, \theta)}{p(Y)}-1\right|^{2sr}\right]\right)^{1/s}\left(\EE\left[ \left|\log \frac{p(Y \,|\, \theta)}{p(Y)}\right|^{2t}\right] +1\right)^{1/t}.
\end{align*}
Since similar bounds exist for the expectations of the first and second terms of \eqref{eq:bound_zl}, we obtain the bound on $\VV[Z_\ell]$ of order $2^{-r\ell}$. A bound on $\EE[|Z_\ell|]$ can be shown similarly.
\end{proof}

\begin{remark}
The result on $\EE[|Z_\ell|]$ implies that the parameter $\alpha$ appearing in Theorem~\ref{thm:basic_MLMC} equals $\min(\frac{p(q-1)}{2q}, 1)$, since
\[ |\EE[P_\ell-P]|=\left|\sum_{\ell'=\ell+1}^{\infty}\EE[Z_{\ell'}]\right|\leq \sum_{\ell'=\ell+1}^{\infty}\EE[|Z_{\ell'}|]=O(2^{-\min(\frac{p(q-1)}{2q}, 1)\ell}). \]
The result on $\VV[|Z_\ell|]$ directly means that the parameter $\beta$ equals $\min(\frac{p(q-2)}{2q}, 2)$. As we have $\gamma=1$, our MLMC estimator is in the regime $\beta>\gamma$ whenever $(p-2)(q-2)> 4$. As a result,  we now know that the MLMC estimator achieves the computational complexity of optimal $O(\varepsilon^{-2})$ for estimating the expected information gain $U$. As mentioned in Remark~\ref{rem:smc}, the standard (nested, in this case) Monte Carlo method only achieves the complexity of $O(\varepsilon^{-2-\gamma/\alpha})$. Since $\alpha=\gamma=1$ whenever $(p-2)(q-1)\geq 2$, we recover the results from \cite{Ryan03,BDELT18}.
\end{remark}

\subsection{Incorporating importance sampling}
In practice, it might be often the case that $p(Y\,|\, \theta)$, as a function of $\theta$ for a fixed $Y$, is highly concentrated around a certain value of $\theta$. If i.i.d.\ random samples of $\theta$ are distributed outside the concentrated region, the Monte Carlo estimates  $\overline{p(Y \,|\, \cdot )}^{(a)}, \overline{p(Y \,|\, \cdot )}^{(b)}$ and $\overline{p(Y \,|\, \cdot )}$ can be numerically zero. This issue is called \emph{arithmetic underflow} \cite{BDELT18}. \tg{This occurs as errors show when numerically taking the logarithm of 0 for $Z_\ell$.} To avoid this issue, we incorporate importance sampling into the MLMC estimator.

Let $q(\theta\,|\, Y)$ be an importance distribution of $\theta$ which satisfies $q(\theta\,|\,Y)>0$ whenever $p(\theta)>0$. For a given $Y$, we have
\[ p(Y)= \EE_\theta[p(Y\,|\, \theta)] = \EE_{\theta\sim q(\cdot\,|\, Y)}\left[ \frac{p(Y\,|\, \theta)p(\theta)}{q(\theta\,|\,Y)}\right], \]
so that the expected information gain $U$ becomes
\[ U = \EE_\theta \left[\EE_{Y|\theta}\left[\log p(Y\,|\, \theta) \right]\right] -\EE_Y\left[\log \EE_{\theta\sim q(\cdot\,|\, Y)}\left[ \frac{p(Y\,|\, \theta)p(\theta)}{q(\theta\,|\,Y)}\right]\right] . \]
The corresponding random variables $P_\ell$ and $Z_\ell$ used in the MLMC estimator are replaced by
\begin{align*}
 \hat{P}_\ell & = \log p(Y \,|\, \theta) - \log \overline{\left(\frac{p(Y \,|\, \cdot )p(\cdot)}{q(\cdot \,|\,Y)}\right)}^{M_\ell}, \\
 \hat{Z}_0  & = \log p(Y \,|\, \theta) - \log \overline{\left(\frac{p(Y \,|\, \cdot )p(\cdot)}{q(\cdot \,|\,Y)}\right)}^{M_0}, \\
 \hat{Z}_\ell & = \frac{1}{2}\left[ \log \overline{\left(\frac{p(Y \,|\, \cdot )p(\cdot)}{q(\cdot \,|\,Y)}\right)}^{(a)}+\log \overline{\left(\frac{p(Y \,|\, \cdot )p(\cdot)}{q(\cdot \,|\,Y)}\right)}^{(b)}\right] - \log \overline{\left(\frac{p(Y \,|\, \cdot )p(\cdot)}{q(\cdot \,|\,Y)}\right)},
\end{align*}
respectively, where the averages are taken with respect to i.i.d.\ random samples of $\theta\sim q(\cdot\,|\, Y)$ for a randomly chosen $Y$.

\begin{remark}
If there exist $p,q>2$ with $(p-2)(q-2)\geq 4$ such that 
\[ \EE_{Y}\EE_{\theta\sim q(\cdot\,|\, Y)}\left[\left|\frac{p(Y\,|\,\theta)p(\theta)}{p(Y)q(\theta\,|\, Y)}\right|^p\right] <\infty \quad \text{and}\quad \EE_{Y}\EE_{\theta\sim q(\cdot\,|\, Y)}\left[\left|\log \frac{p(Y\,|\,\theta)p(\theta)}{p(Y)q(\theta\,|\, Y)}\right|^q\right]<\infty , \]
respectively, a similar proof to that of Theorem~\ref{thm:main} goes through and we obtain
\[ \EE[|\hat{Z}_\ell|] =O(2^{-\min(\frac{p(q-1)}{2q}, 1)\ell})\quad \text{and}\quad \VV[\hat{Z}_{\ell}]=O(2^{- \min(\frac{p(q-2)}{2q}, 2) \ell}). \]
Hence the MLMC estimator with importance sampling still achieves the computational complexity of $O(\varepsilon^{-2})$ whenever $(p-2)(q-2)> 4$.
\end{remark}

The question is how to construct an importance distribution $q(\theta\,|\, Y)$ depending on each particular problem. The common guideline is to find a good approximation of the posterior distribution $p(\theta\,|\,Y)$. The Laplace approximation method, which has been recently studied in \cite{LSTW13,BDELT18} for estimating the expected information gain, is a method to approximate $p(\theta\,|\,Y)$ by a (multivariate) Gaussian distribution, When the data $Y$ is generated conditionally on the known value of $\theta=\theta^*$ \tg{from \eqref{eq:data_model}}, the Laplace method approximates $p(\theta\,|\, Y)$ by a Gaussian distribution $N(\hat{\theta}, \hat{\Sigma})$, for instance, with
\begin{align*}
\hat{\theta} & = \theta^* - \left( J(\theta^*)^{\top}\Sigma_\epsilon^{-1} J(\theta^*)+H^{\top}(\theta^*)\Sigma_\epsilon^{-1}E -\nabla_\theta\nabla_\theta \log(p(\theta^*)) \right)^{-1}J(\theta^*)^{\top}\Sigma_\epsilon^{-1} E, \\ 
\hat{\Sigma} & = \left( J(\hat{\theta})^{\top}\Sigma_\epsilon^{-1} J(\hat{\theta})-\nabla_\theta\nabla_\theta \log(p(\hat{\theta})) \right)^{-1}.
\end{align*}
Here we \tg{denote the Jacobian and Hessian of $-g$ by $J$ and $H$, respectively, that is, $J(\theta):=-\nabla_\theta g(\theta)$, $H(\theta):=-\nabla_\theta\nabla_\theta g(\theta)$, and moreover we write $E := Y-g(\theta^*)$}. We refer to \cite{LSTW13,BDELT18} for details. It is clear that we need to compute the first-order and second-order derivatives of $g$ with respect to $\theta$. \tg{Typically when} their analytical computations are not available, we may approximate them by finite differences.

\section{Numerical experiments}\label{sec:numer}
Two examples are presented here to demonstrate the efficiency of our MLMC estimator by comparing the numerical performance with that of the NMC estimator. In order to avoid arithmetic underflow, we always use the Laplace-based importance sampling within both the MLMC and the NMC estimators. The first example is a simple test case where the analytical value of $U$ is available, while the second one is based on a more realistic pharmacokinetic (PK) model adapted from \cite{RDTP14}. Throughout all the experiments, we set $M_0$ (the number of inner samples at level 0) to be 1. 

\subsection{Simple test case}
Let $\theta$ be a vector in $\RR^{\tg{d}}$ and consider the following linear data model:
\[ Y=A\theta +\epsilon, \]
where $A\in \RR^{\tg{w\times d}}$ and $Y, \epsilon \in \RR^{\tg{w}}$. We assume that the prior distribution of $\theta$ is given by the multivariate Gaussian distribution $N(\mu_\theta, \Sigma_\theta)$ and the noise $\epsilon$ follows $N(\bszero, \Sigma_\epsilon)$. Allowing to repeat experiments $N_e$ times, the expected information gain for this model can be evaluated analytically as
\[ U = \frac{1}{2}\log |N_e \Sigma_\epsilon^{-1}A\Sigma_\theta A^{\top}+I|, \]
where $I$ denotes the identity matrix of size \tg{$w\times w$}.

In what follows, we set $\tg{d=2, w=3}, \mu_\theta=(1,0)^{\top},$
\[ \Sigma_\theta =  \begin{bmatrix} 2 & -1 \\ -1 & 2 \\ \end{bmatrix},\quad  A =  \begin{bmatrix} 1 & 2 \\ 2 & 3 \\ 3 & 4 \\ \end{bmatrix}, \quad \text{and}\quad \Sigma_\epsilon =  \begin{bmatrix} 0.1 & -0.05 & 0 \\ -0.05 & 0.1 & -0.05 \\ 0 & -0.05 & 0.1 \\ \end{bmatrix}. \]
For this parameter setting, the analytical values of $U$ for the cases $N_e=1$ and $N_e=10$ are $4.4574$ and $6.6642$, respectively.

The numerical result for the case $N_e=1$ is shown in Fig.~\ref{fig:test1}. The left top plot shows the behaviors of the mean values of both $P_\ell$ and $Z_\ell$, where the means are estimated empirically by using $2\times 10^4$ random samples for each level. Note that the logarithm of the absolute mean value in base $2$ is plotted as a function of level. While the mean value of $P_\ell$ is almost constant, the absolute mean value of $Z_\ell$ decays geometrically fast as the level increases. The slope of the line for $Z_\ell$ is $-0.93$, which means $\alpha=0.93$ and is in good agreement with Theorem~\ref{thm:main}.

The right top plot shows the behaviors of the empirical variances of both $P_\ell$ and $Z_\ell$. Here we again plot the logarithm of the variance in base $2$ as a function of level. While the variance of $P_\ell$ is almost constant, the variance of $Z_\ell$ decays geometrically fast as the level increases. The slope of the line for $Z_\ell$ is $-1.64$, which means $\beta=1.64$ and again agrees well with Theorem~\ref{thm:main}. These two convergence results in conjunction with the fact $\gamma=1$ indicate that the MLMC estimator can achieve the computational complexity of $O(\varepsilon^{-2})$ for estimating $U$.

In order to confirm that this indication is indeed the case in practice, we run the following algorithm which is a slight modification from one described in \cite[Section~3.1]{G15}.
\begin{algorithm}\label{alg:MLMC}
Let $\omega\in (0,1)$ be a user-specified parameter. For a target root-mean-square accuracy $\varepsilon$, start with $L=L_0$ and give an initial number of samples $N_*$ for all the levels $\ell=0,\ldots,L$. Until extra samples need to be evaluated, repeat the following:
\begin{enumerate}
\item evaluate extra samples on each level.
\item compute (or update) the empirical variances $\hat{V}_\ell$ for $\ell=0,\ldots,L$. 
\item define optimal $N_\ell$ for $\ell=0,\ldots,L$ according to
\[ N_\ell = \left\lceil (1-\omega)^{-1}\varepsilon^{-2}\sqrt{\frac{\hat{V}_\ell}{C_\ell}}\sum_{\ell=0}^{L}\sqrt{\hat{V}_\ell C_\ell}\right\rceil. \]
\item test for the bias convergence $|\EE[Z_L]|/(2^{\alpha}-1) \leq \sqrt{\omega}\varepsilon$, where we use the empirical estimates for $\EE[Z_\ell]$ and $\alpha$.
\item if the bias is not converged, let $L=L+1$ and give an initial number of samples $N_L$. 
\end{enumerate}
\end{algorithm}

In this algorithm, the optimal allocation of $N_\ell$ given in Item~3 is derived by minimizing the total cost $\sum_{\ell=0}^{L}N_\ell C_\ell$ for a fixed variance $\sum_{\ell=0}^{L}V_\ell/N_\ell = (1-\omega)\varepsilon^2$. The bias convergence test in Item~4 comes from the assumption $\EE[Z_\ell]\propto 2^{-\alpha \ell}$, which leads to
\[ \EE[P-P_L] = \sum_{\ell=L+1}^{\infty}\EE[Z_\ell] = \frac{\EE[Z_L]}{2^{\alpha}-1}. \]
In this way, Algorithm~\ref{alg:MLMC} heuristically ensures that the mean square error \eqref{eq:MLMC_mse} of the MLMC estimator is bounded above by
\[ \EE[( Z_{\text{MLMC}}-\EE[P])^2] =\sum_{\ell=0}^{L}\frac{V_\ell}{N_\ell}+\left( \EE[P_L-P]\right)^2 \leq  (1-\omega)\varepsilon^2+\omega\varepsilon^2=\varepsilon^2. \]
In our experiments, we always put $\omega=0.25, L_0=2$ and $N_*=10^3$.

\begin{figure}[t]
\centering
\includegraphics[width=0.8\textwidth]{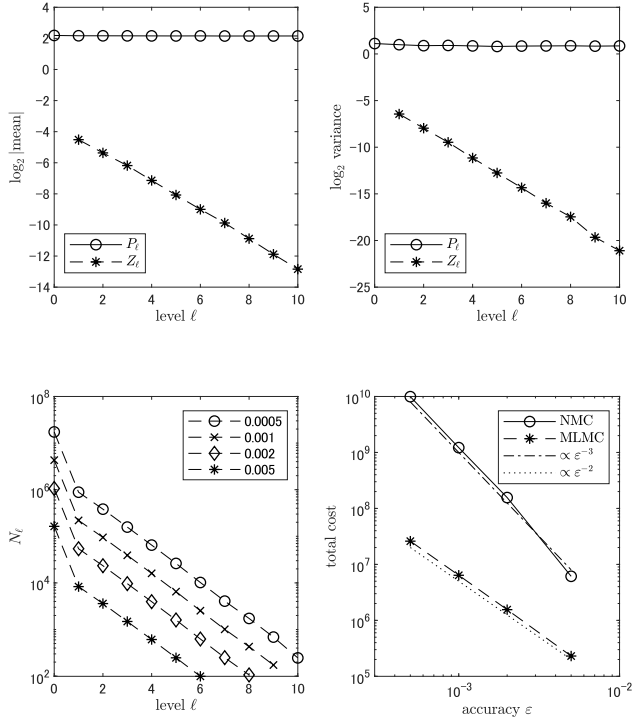}
\caption{Numerical results for the test case with $N_e=1$.}
\label{fig:test1}
\end{figure}
The left bottom plot of Fig.~\ref{fig:test1} shows the resulting allocation of $N_\ell$ from $\ell=0$ to the maximum level $\ell=L$ for different values of $\varepsilon$. We see that, as $\varepsilon$ decreases, the maximum level $L$ increases so as to satisfy the bias convergence. As expected, for any $\varepsilon$, $N_\ell$ decreases geometrically as the level increases, i.e., most of the samples are allocated on the coarser levels. The right bottom plot compares the total cost required for the MLMC estimator to have the root-mean-square accuracy less than $\varepsilon$ with that for the NMC estimator. Here the total cost for the NMC estimator is computed by
\[ C_L\times \frac{\hat{\VV}[P_L]}{(1-\omega)\varepsilon^2},\]
for the same maximum level $L$ with the MLMC estimator, so that the mean square error \eqref{eq:SMC_mse} of the NMC estimator is bounded above by $\varepsilon^2$. As the theoretical result predicted, we see that the total cost of the MLMC estimator is of $O(\varepsilon^{-2})$, whereas that of the NMC estimator is of $O(\varepsilon^{-3})$. For $\varepsilon=5\times 10^{-4}$, the MLMC estimator is more than 380 times more efficient. The estimated $U$ is $4.458$, which agrees quite well with the analytical value.

\begin{figure}[t]
\centering
\includegraphics[width=0.8\textwidth]{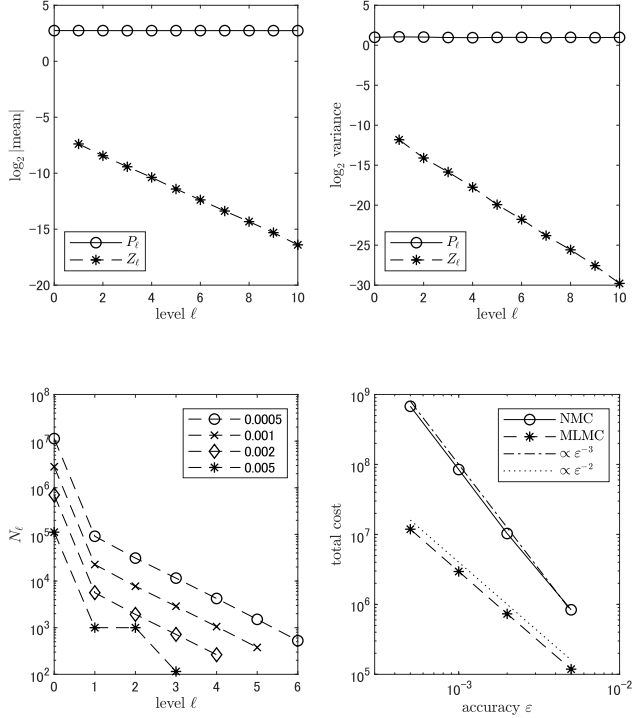}
\caption{Numerical results for the test case with $N_e=10$.}
\label{fig:test2}
\end{figure}
As shown in Fig.~\ref{fig:test2}, even for the case $N_e=10$, similar convergence behaviors of the mean value $|\EE[Z_\ell]|$ and the variance $\VV[Z_\ell]$ are observed. In this case, the estimated values of $\alpha$ and $\beta$ are $0.99$ and $1.97$, respectively. For $\varepsilon=5\times 10^{-4}$, the MLMC estimator achieves the computational saving of a factor more than $50$. The estimated $U$ is $6.664$, which again agrees well with the analytical value.

\subsection{Pharmacokinetic model}
Let us consider a more realistic example which is adapted from the PK model used in \cite[Example~3]{RDTP14}. Suppose that a drug is administrated to subjects. In order to reduce the uncertainty about a set of PK parameters, which affect the absorption, distribution and elimination of the drug in the subjects' body, it would be helpful to take blood samples of the subjects at several different times and to measure the concentration of drug in the samples.

In the data model \eqref{eq:data_model} considered in this paper, $\theta$ is a set of PK parameters, $\xi$ is a set of blood sampling times after the administration of the drug, denoted by $\xi=(t_1,\ldots,t_J)$, and $Y=(Y_1,\ldots,Y_J)$ is a vector of the measured drug concentration at times $t=t_1,\ldots,t_J$. Following \cite{RDTP14},  let $\theta=(k_a, k_e,V)\in \RR_{>0}^3$ with $k_a$ being the first-order absorption constant, $k_e$ the first-order elimination constant, $V$ the volume of distribution. The drug concentration at time $t_j$, where hour is used as a unit, is modeled as
\[ Y_j = g_{t_j}(k_a, k_e,V)+ \epsilon := \frac{D}{V}\frac{k_a}{k_e-k_a}\left( e^{-k_e t_j}- e^{-k_a t_j}\right) + \epsilon, \]
with the white noise $\epsilon\sim N(0,0.01)$ \tg{and with a single fixed dose $D=400$ administrated at the beginning of the experiment.} The difference from the original model in \cite{RDTP14} is that we remove one noise term whose variance depends on the value of $g_{t_j}$ for simplicity. The prior probability distributions of $k_a, k_e$ and $V$ are independent and given by $\log k_a\sim N(0, 0.05)$, $\log k_e\sim N(\log 0.1, 0.05)$ and $\log V\sim N(\log 20, 0.05)$, respectively. Regarding the experimental setup $\xi$, we follow \cite{RDTP14} and consider three different blood sampling schemes with all $J=15$:
\begin{enumerate}
\item (beta) $\xi_1$: Percentiles of the Beta$(0.7, 1.2)$ distribution, scaled to $[0, 24]$,
\item (even-spacing) $\xi_2$: $t_j=0.3+1.6 \times (j-1)$,
\item (geometric) $\xi_3$: $t_j=0.94\times 1.25^{j-1}$.
\end{enumerate}

\begin{figure}[t]
\centering
\includegraphics[width=0.8\textwidth]{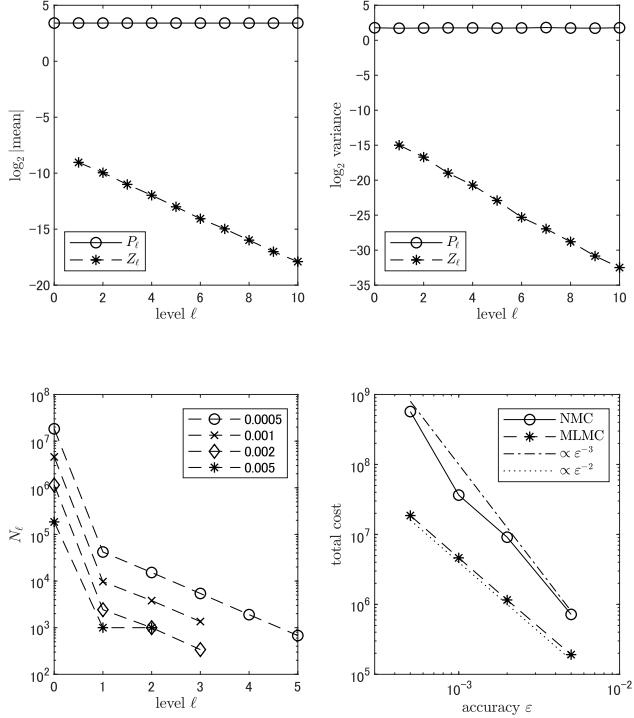}
\caption{Numerical results for the PK model with the beta-scheme sampling times.}
\label{fig:pk_beta}
\end{figure}

\begin{figure}[t]
\centering
\includegraphics[width=0.8\textwidth]{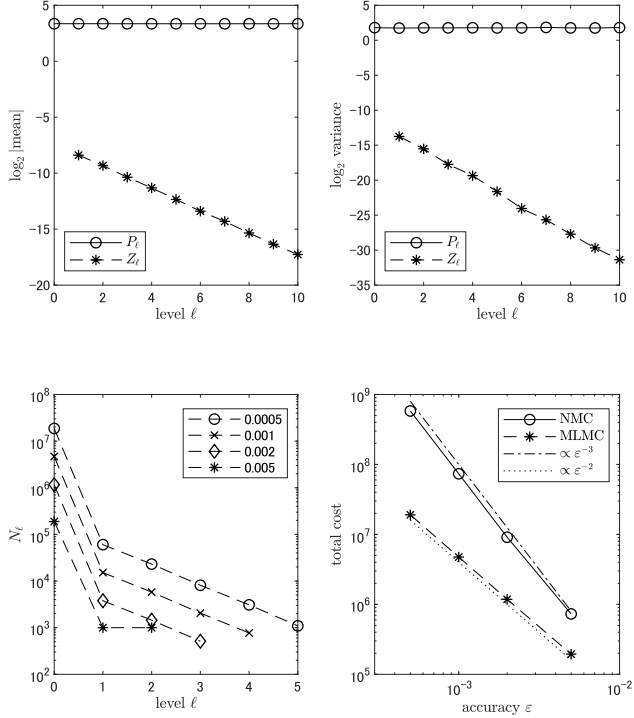}
\caption{Numerical results for the PK model with the even-spacing-scheme sampling times.}
\label{fig:pk_even}
\end{figure}

\begin{figure}[t]
\centering
\includegraphics[width=0.8\textwidth]{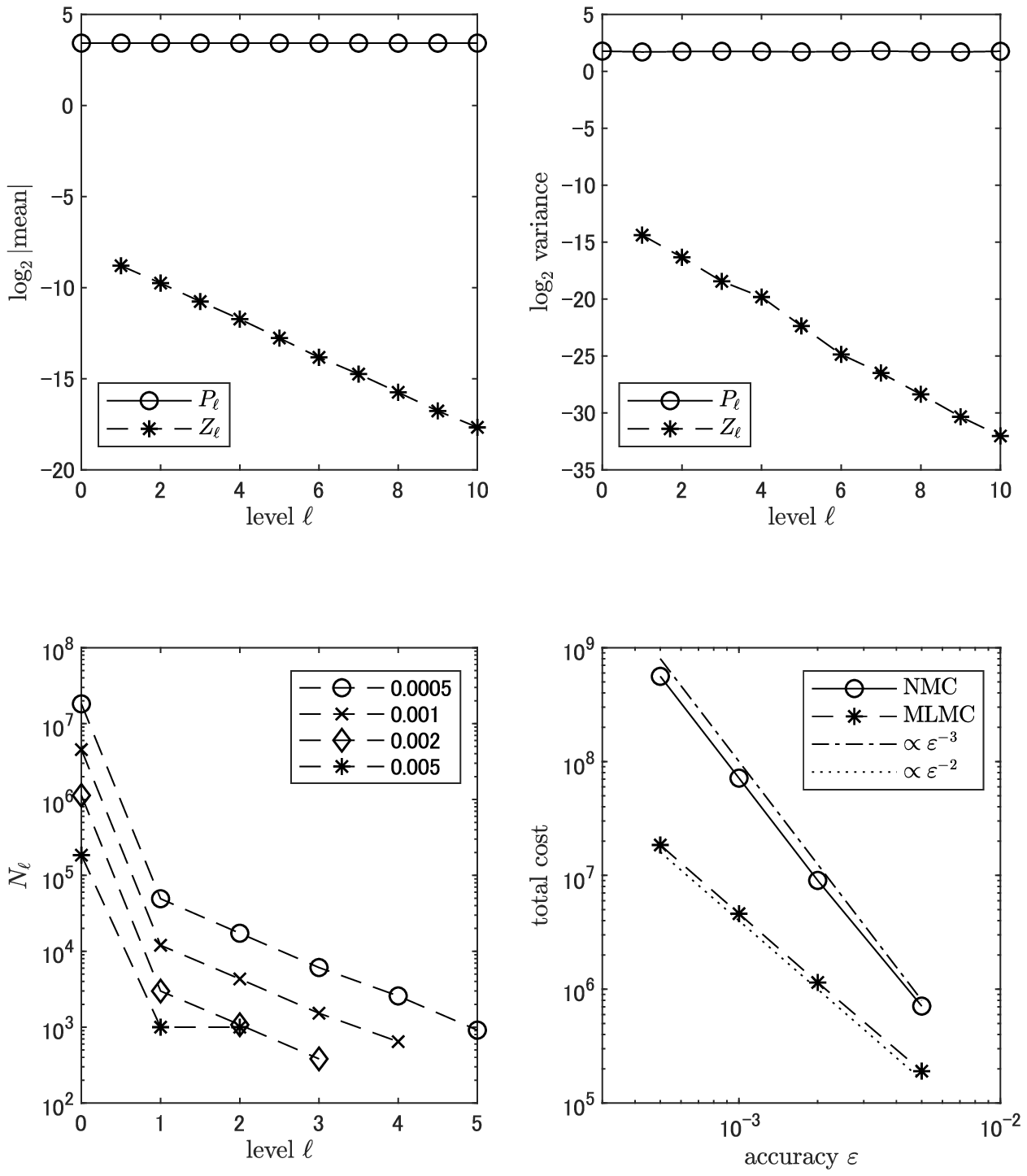}
\caption{Numerical results for the PK model with the geometric-scheme sampling times.}
\label{fig:pk_geom}
\end{figure}

Figs.~\ref{fig:pk_beta}--\ref{fig:pk_geom} show the MLMC numerical results for three respective experimental setups. For any setup, we can see the geometric decay of both $|\EE[Z_\ell]|$ and $\VV[Z_\ell]$, which confirms the tight coupling of the corrections $Z_\ell$. Similarly to the simple test case, the total cost for the MLMC estimator is of $O(\varepsilon^{-2})$, whereas that for the NMC estimator is of $O(\varepsilon^{-3})$. In Table~\ref{tbl:summary_PK}, we summarize these results. In our problem setting, the expected information gain for the geometric-scheme sampling $\xi_3$ is slightly larger than that for the beta-scheme sampling $\xi_1$, which itself is larger than that for the even-spacing-scheme sampling $\xi_2$. Thus $\xi_3$ is the best experimental setup among these three. There may exist a better experimental setup yielding a larger $U$, although such an investigation is the beyond the scope of this paper.

\begin{table}[h]
\caption{Summery of numerical results for the PK model}\label{tbl:summary_PK}\tg{
\begin{tabular}{c||c|c|c|c|c|c}
\hline
Sampling scheme & $\alpha$ & $\beta$ & MLMC cost & NMC cost & saving & $U$ \\ \hline
beta & 0.995 & 1.980 & $1.86\times 10^7$ & $5.70\times 10^8$ & 30.6 & 10.63 \\
even-spacing & 0.994 & 1.993 & $1.91\times 10^7$ & $5.81\times 10^8$ & 30.5 & 10.21 \\
geometric & 0.994 & 1.994 & $1.85\times 10^7$ & $5.60\times 10^8$ & 30.3 & 10.74 \\ \hline
\end{tabular}}
The MLMC cost, the NMC cost and the saving are the results for $\varepsilon=5\times 10^{-4}$.
\end{table}
\section{Conclusion}\label{sec:conclusion}
In this paper we have developed an MLMC estimator for the expected information gain, which is one of the most important quality criteria of Bayesian experimental designs. Under a sufficient condition on the data model, we prove that our MLMC estimator achieves the computational complexity of $O(\varepsilon^{-2})$, which compares favorably with that of the nested Monte Carlo estimator, which is $O(\varepsilon^{-3})$. Combining importance sampling techniques with the MLMC estimator is straightforward and is quite helpful not only in reducing the variance of the corrections $Z_\ell$ but also, as shown in \cite{BDELT18}, in mitigating the risk of arithmetic underflow. Numerical experiments support our theoretical result.

We leave the following issues open for future research:
\begin{itemize}
\item an extension to the situation where the function $g$ can only be evaluated approximately. As studied in \cite{BDELT18}, in some engineering applications, we have to deal with the situation where $g$ is a functional of the solution of partial differential equations and only approximate values of $g$ from finite difference or finite element approximations are available. Soon after completing the first version of this paper, an independent work by Beck et al.\ \cite{BDETxx} has introduced the MLMC estimator of the expected information gains for such situations.\footnote{The authors used the standard (non-antithetic) MLMC estimator and claimed that the property $\beta=2$ holds without a rigorous argument. However, this present work supports this claim theoretically if we use the antithetic MLMC estimator.} As a natural extension, a multi-index Monte Carlo method \cite{HNT16} can be considered to improve the computational efficiency.
\item the use of quasi-Monte Carlo (QMC) sampling instead of i.i.d.\ random sampling. The idea behind QMC sampling is that by distributing samples more uniformly or evenly over the domain, i.e., by generating ``low-discrepancy'' points or sequences, the rate of convergence for estimating expectations is to be improved. There are some works which combine QMC sampling with MLMC, see for instance \cite{GW09,GMTS18}. It is expected to achieve additional computational savings also in the current application.
\item a combination with an optimization algorithm to find optimal Bayesian experimental designs. The ultimate goal in this direction of research would be to efficiently construct optimal Bayesian experimental designs. In this paper, we only dealt with an estimation of the expected information gain for a given experimental setup. Combining the MLMC estimator with an optimization algorithm would be a promising approach to attain this goal.
\end{itemize}

\end{document}